\documentclass[a4paper,twoside,english,DIV12]{scrartcl}
\usepackage[T1]{fontenc}
\usepackage{verbatim}
\usepackage{mathtools}
\usepackage{amsthm}
\usepackage{amsmath}
\usepackage{amssymb}
\usepackage[numbers]{natbib}

\makeatletter

\pdfpageheight\paperheight
\pdfpagewidth\paperwidth

\newcommand{\noun}[1]{\textsc{#1}}

\theoremstyle{plain}
\newtheorem{thm}{\protect\theoremname}[section]
  \theoremstyle{definition}
  \newtheorem{defn}[thm]{\protect\definitionname}
  \theoremstyle{remark}
  \newtheorem{rem}[thm]{\protect\remarkname}
  \theoremstyle{plain}
  \newtheorem{cor}[thm]{\protect\corollaryname}

\usepackage[english]{babel}     % Deutsche Silbentrennung
\usepackage[T1]{fontenc}    % Silbentrennung auch nach Umlauten mglich

\usepackage{a4wide}        % A4-Format
\usepackage{amsmath}
\usepackage{euscript}
\usepackage{url}
\usepackage{scrpage2}
\usepackage{tu-preprint}
\allowdisplaybreaks[1] % allow page breaks in some displayed equations
\usepackage{amsfonts}

\newcommand{\hide}[1]{}

\newcommand*{\dive}{\operatorname{div}}

\newcommand*{\Grad}{\operatorname{Grad}}

\newcommand*{\Dive}{\operatorname{Div}}
\newcommand*{\sym}{\operatorname{sym}}

\newcommand*{\grad}{\operatorname{grad}}

\newcommand*{\ii}{\mathrm{i}}
\DeclareMathAccent{\Circ}{\mathalpha}{operators}{"17}
\newcommand{\interior}[1]{\Circ{#1}}
\renewcommand{\Im}{\operatorname{\mathfrak{Im}}}
\renewcommand{\Re}{\operatorname{\mathfrak{Re}}}

\newcommand{\oi}[2]{\left]#1,#2 \right[}
\newcommand{\rga}[1]{\left]#1,\infty  \right[}

\renewcommand*{\epsilon}{\varepsilon}
\renewcommand*{\rho}{\varrho}

\institut{Institut f\"ur Analysis}

\preprintnumber{MATH-AN-04-2014}

\preprinttitle{On Some Models in Linear Thermo-Elasticity with Rational Material Laws.}

\author{S. Mukhopadhyay, R. Picard, S. Trostorff, M. Waurick}

\makeatother

\usepackage{babel}
  \providecommand{\corollaryname}{Corollary}
  \providecommand{\definitionname}{Definition}
  \providecommand{\remarkname}{Remark}
\providecommand{\theoremname}{Theorem}

\begin{document}
\makepreprinttitlepage

\title{On Some Models in Linear Thermo-Elasticity with Rational Material
Laws.}

\author{S. Mukhopadhyay%
\thanks{\emph{S. Mukhopadhyay:} Department of Mathematical Sciences, Indian
Institute of Technology (BHU), Varanasi -221005, India: Corresponding
author: E-mail: mukhosant.apm@iitbhu.ac.in %
}, R. Picard\textit{\emph{,}} S. Trostorff, M. Waurick\textit{\emph{}}%
\thanks{\emph{R. Picard}\textit{\emph{,}}\emph{ S. Trostorff, M. Waurick}\textit{\emph{:
Institute for Analysis, Faculty of Mathematics and Sciences, TU-Dresden,
Germany.}}%
}}

\maketitle
\setcounter{section}{-1}

\date{}

\textbf{\textit{Abstract: }}We shall consider some common models in
linear thermo-elasticity within a common structural framework. Due
to the flexibility of the structural perspective we will obtain well-posedness
results for a large class of generalized models allowing for more
general material properties such as anisotropies, inhomogeneities,
etc.

\section{Introduction}

The coupled dynamical thermoelasticity (CTE) theory was developed
by Biot \citep*{Biot1956} to eliminate the drawback of uncoupled
theory of thermoelasticity that the elastic changes in a material
have no effects on temperature. Like other classical thermodynamical
theories of continua, this theory is developed on the basis of firm
grounds of irreversible thermodynamics by employing Fourier's law
and has been used to study the coupling effects of elastic and thermal
fields over the years. However, this theory suffers from the paradox
of infinite heat propagation speed and predicts unsatisfactory descriptions
of a solid\textquoteright s response to some situations, like fast
transient loading at low temperature, etc. Generalized thermoelasticity
theories are therefore developed in last few decades with the aim
to eliminate this drawback. Extended thermoelasticity (ETE) theory
was introduced by Lord and Shulman \citep{Lord1967} by employing
a modified Fourier law proposed by Catteneo \citep{cattaneo1958}
and Vernotte \citep{Vernotte1958,Vernotte1961} that includes one
thermal relaxation time parameter. Temperature-rate dependent thermoelasticity
(TRDTE) theory by Green and Lindsay \citep{Green_Lindsay1972} and
thermoelasticity theories of type I, II and III by Green and Naghdi
\citep{Green_Naghdi1991,Green_Naghdi1992,Green_Naghdi1993} are also
advocated in this context. Later on, Chandrasekharaiah \citep{Chandra1998}
modified the governing equations of thermoelasticity on the basis
of a so-called dual phase-lag heat conduction equation due to Tzou
\citep{Tzou1995,Tzou1995_2} and proposed two different models of
thermoelasticity, namely dual phase-lag model- I (DPL-I) and dual
phase-lag model- II (DPL-II). The dual phase lag heat conduction law
is supposed to be the macroscopic formulation of the microscopic effects
in heat transport processes. A possible application of this generalized
heat conduction law arises in the modeling of laser pulses. It has
been found out that laser pulses can be shortened to the range of
femtoseconds ($10^{-15}$ s). When the response time becomes shorter,
the non-equilibrium thermodynamic transition and the microscopic effects
in the energy exchange during heat transport procedure become pronounced
(Tzou \citep{tzou1996macro}). The formulation therefore becomes microscopic
in nature. The dual phase-lag heat conduction law incorporates this
microscopic effects in heat transport process by introducing two macroscopic
lagging (or delayed) responses as possible outcomes. A detailed history
about the development of some well- established non-Fourier heat conduction
models and their importance are available in the references \citep{Chandra1986,Chandra1998,Hetnarski1999,hetnarski2008thermal,Ignaczak1989,Ignaczak1991,Ignaczak2010,Joseph1989,Quintanilla2006,Kothari2013,mukhopadhyay2014dual}. 

Recently, a structural formulation for linear material laws in classical
mathematical physics was reported by Picard \citep{Picard2009}. Here,
a class of evolutionary problems is considered to cover a number of
initial boundary value problems of classical mathematical physics
and the solution theory is established. The well-posedness of classical
thermoelasticity and Lord Shulman theory was shown to be covered by
this model. The main objective of this present work is to show that
the aforementioned models of generalized thermoelasticity can be treated
within the common structural framework of evolutionary equations.
Due to the flexibility of the structural perspective we will obtain
well-posedness results for a large class of generalized models allowing
for more general material properties such as anisotropies, inhomogeneities,
etc. The solution strategy is not based on constructions involving
fundamental solutions (semi-group theory), which will allow for even
more general materials resulting for example in changes of type (e.g.
from parabolic to hyperbolic) or for suitable non-local material properties
involving e.g. spatial integral operators. It should be noted that
evolutionary equations in the form just discussed have also been studied
with regards to homogenization theory, see e.g.~\citep{Waurick2012,Waurick2012a,Waurick2013}.
Hence, the general perspective on thermo-elasticity to be presented
may also shed some new light on the theory of homogenization of such
models.

The article is structured as follows. We begin to introduce the framework
of evolutionary equations and recall the general well-posedness result.
We will focus on so-called rational material laws defined as functions
of the time-derivative $\partial_{0}$, which is established as a
normal operator in a suitable exponentially weighted $L^{2}$-space.
In the proceeding sections we will show, how the different models
of generalized elasticity can be incorporated into this framework
and we will derive assumptions on the material coefficients yielding
the well-posedness of the corresponding evolutionary equations.

\section{Foundations}

\subsection{The framework of evolutionary equations}

The family of Hilbert spaces $\left(H_{\rho,0}\left(\mathbb{R},H\right)\right)_{\rho\in\mathbb{R}}$,
$H$ complex Hilbert space, with $H_{\rho,0}\left(\mathbb{R},H\right)\coloneqq L^{2}\left(\mathbb{R},\mu_{\rho},H\right)$,
where the measure $\mu_{\rho}$ is defined by $\mu_{\rho}\left(S\right)\coloneqq\int_{S}\exp\left(-2\rho t\right)\:\mathrm{d}t$,
$S\subseteq\mathbb{R}$ a Borel set, $\rho\in\mathbb{R}$, provides
the desired Hilbert space setting for evolutionary problems (cf. \citep{Picard_PDE,Kalauch2014}).
The sign of $\rho$ is associated with the direction of causality,
where the positive sign is linked to forward causality. Since we have
a preference for forward causality, we shall usually assume that $\rho\in\rga0$.
By construction of these spaces, we can establish
\begin{align*}
\exp\left(-\rho\mathbf{m}_{0}\right):H_{\rho,0}\left(\mathbb{R},H\right) & \to H_{0,0}\left(\mathbb{R},H\right)(=L^{2}\left(\mathbb{R},H\right))\\
\varphi & \mapsto\exp\left(-\rho\mathbf{m}_{0}\right)\varphi
\end{align*}
where $\left(\exp\left(-\rho\mathbf{m}_{0}\right)\varphi\right)\left(t\right)\coloneqq\exp\left(-\rho t\right)\varphi\left(t\right)$,
$t\in\mathbb{R}$, as a unitary mapping. We use $\mathbf{m}_{0}$
as a notation for the multiplication-by-argument operator corresponding
to the time parameter.

In this Hilbert space setting the time-derivative operation, defined
as the closure of 
\begin{align*}
\interior C_{\infty}(\mathbb{R},H)\subseteq H_{\rho,0}(\mathbb{R},H) & \to H_{\rho,0}(\mathbb{R},H)\\
\varphi & \mapsto\dot{\varphi},
\end{align*}
where by $\interior C_{\infty}(\mathbb{R},H)$ we denote the space
of arbitrary differentiable functions from $\mathbb{R}$ to $H$ having
compact support, generates a normal operator $\partial_{0,\rho}$
with%
\footnote{Recall that for normal operators $N$ in a Hilbert space $H$
\begin{align*}
\Re N & \coloneqq\frac{1}{2}\overline{\left(N+N^{*}\right)},\\
\Im N & \coloneqq\frac{1}{2\ii}\overline{\left(N-N^{*}\right)}
\end{align*}
and
\[
N=\Re N+\ii\Im N.
\]
It is
\[
D\left(N\right)=D\left(\Re N\right)\cap D\left(\Im N\right).
\]
}
\begin{align*}
\Re\partial_{0,\rho} & =\rho,\\
\Im\partial_{0,\rho} & =\frac{1}{\ii}\left(\partial_{0,\rho}-\rho\right).
\end{align*}
The skew-selfadjoint operator $\ii\Im\partial_{0,\rho}$ is unitarily
equivalent to the differentiation operator $\partial_{0,0}$ in $L^{2}\left(\mathbb{R},H\right)=H_{0}\left(\mathbb{R},H\right)$
with domain $H^{1}(\mathbb{R},H)$ - the space of weakly differentiable
functions in $L^{2}(\mathbb{R},H)$ - via
\[
\ii\Im\partial_{0,\rho}=\left(\exp\left(-\rho\mathbf{m}_{0}\right)\right)^{-1}\partial_{0,0}\exp\left(-\rho\mathbf{m}_{0}\right)
\]
and has the Fourier-Laplace transformation as its spectral representation,
which is the unitary transformation
\[
\mathcal{L}_{\rho}\coloneqq\mathcal{F}\:\exp\left(-\rho\mathbf{m}_{0}\right):H_{\rho,0}(\mathbb{R},H)\to L^{2}(\mathbb{R},H),
\]
where $\mathcal{F}:L^{2}(\mathbb{R},H)\to L^{2}(\mathbb{R},H)$ is
the Fourier transformation given as the unitary extension of 
\[
\interior C_{\infty}(\mathbb{R},H)\ni\varphi\mapsto\left(s\mapsto\frac{1}{\sqrt{2\pi}}\intop_{\mathbb{R}}\exp(-\ii s\cdot t)\varphi(t)\,\mathrm{d}t\right).
\]
Indeed, this follows from the well-known fact that $\mathcal{F}$
is unitary in $L^{2}\left(\mathbb{R},H\right)$ and a spectral representation
for $\frac{1}{\ii}\partial_{0,0}$ in $L^{2}\left(\mathbb{R},H\right)$.
In particular, we have 
\begin{align*}
\Im\partial_{0,\rho} & =\mathcal{L}_{\rho}^{*}\mathbf{m}_{0}\mathcal{L}_{\rho}
\end{align*}
and thus, 
\[
\partial_{0,\rho}=\mathcal{L}_{\rho}^{*}\left(\ii\mathbf{m}_{0}+\rho\right)\mathcal{L}_{\rho}.
\]

It is crucial to note that for $\rho\not=0$ we have that $\partial_{0,\rho}$
has a bounded inverse. Indeed, for $\rho>0$ we find from $\Re\partial_{0,\rho}=\rho$
that
\begin{equation}
\left\Vert \partial_{0,\rho}^{-1}\right\Vert _{\rho,0}=\frac{1}{\rho},\label{eq:time-integral-norm-1}
\end{equation}
where $\left\Vert \:\cdot\;\right\Vert _{\rho,0}$, $\rho\in\rga0$
denotes the operator norm on $H_{\rho,0}(\mathbb{R},H)$. For continuous
functions $\varphi$ with compact support we find
\begin{equation}
\left(\partial_{0,\rho}^{-1}\varphi\right)\left(t\right)=\int_{-\infty}^{t}\varphi\left(s\right)\:\mathrm{d}s,\: t\in\mathbb{R},\rho\in\oi{0}{\infty},\label{eq:integral-1}
\end{equation}

which shows the causality of $\partial_{0,\rho}^{-1}$ for $\rho>0$.%
\footnote{If $\rho<0$ the operator $\partial_{0,\rho}$ is also boundedly invertible
and its inverse is given by 
\[
\left(\partial_{0,\rho}^{-1}\varphi\right)(t)=-\intop_{t}^{\infty}\varphi(s)\, ds\quad(t\in\mathbb{R})
\]
for all $\varphi\in\interior C_{\infty}(\mathbb{R},H)$. Thus, $\rho<0$
corresponds to the backward causal (or anticausal) case.%
} Since it is usually clear from the context which $\rho$ has been
chosen, we shall, as it is customary, drop the index $\rho$ from
the notation for the time derivative and simply use $\partial_{0}$
instead of $\partial_{0,\rho}$.\\
We are now able to define operator-valued functions of $\partial_{0}$
via the induced function calculus of $\partial_{0}^{-1}$ as
\[
M\left(\partial_{0}^{-1}\right)\coloneqq\mathcal{L}_{\rho}^{*}M\left(\left(\ii\mathbf{m}_{0}+\rho\right)^{-1}\right)\mathcal{L}_{\rho}.
\]
Here, we require that $z\mapsto M\left(z\right)$ is a bounded, analytic
function defined on $B_{\mathbb{C}}\left(\frac{1}{2\rho_{0}},\frac{1}{2\rho_{0}}\right)$
for some $\rho_{0}\in\oi0\infty$ attaining values in $L(H)$, the
space of bounded linear operators on $H$. Then, for $\rho>\rho_{0}$
the operator $M\left(\left(\ii\mathbf{m}_{0}+\rho\right)^{-1}\right)$
defined as 
\[
\left(M\left(\left(\ii\mathbf{m}_{0}+\rho\right)^{-1}\right)f\right)(t)\coloneqq M\left(\left(\ii t+\rho\right)^{-1}\right)f(t)\quad(t\in\mathbb{R},f\in L^{2}(\mathbb{R},H))
\]
is bounded and linear, and hence $M(\partial_{0}^{-1})\in L(H_{\rho,0}(\mathbb{R},H)).$
Moreover, due to the analyticity of $M$ we obtain that $M(\partial_{0}^{-1})$
becomes causal (see \citep[Theorem 2.10]{Picard2009}). \\
We recall from \citep{Picard2009} (and the concluding chapter of
\citep{Picard_PDE}) that the common form of standard initial boundary
value problems of mathematical physics is given by
\begin{equation}
\overline{\left(\partial_{0}M\left(\partial_{0}^{-1}\right)+A\right)}U=F,\label{eq:evo2-1}
\end{equation}
where $A$ is the canonical skew-selfadjoint extension to $H_{\rho,0}\left(\mathbb{R},H\right)$
of a skew-selfadjoint operator in $H$. We recall the well-posedness
result for this class of problems.
\begin{thm}[{\citep[Solution Theory]{Picard2009}}]
\label{thm:sol_th} Let $A:D(A)\subseteq H\to H$ be a skew-selfadjoint
operator and $M:B\left(\frac{1}{2\rho_{0}},\frac{1}{2\rho_{0}}\right)\to L(H)$
an analytic and bounded mapping, where $\rho_{0}\in\oi0\infty.$ Assume
that there is $c\in\oi0\infty$ such that for all $z\in B\left(\frac{1}{2\rho_{0}},\frac{1}{2\rho_{0}}\right)$
the estimate 
\begin{equation}
\Re z^{-1}M(z)=\frac{1}{2}\left(z^{-1}M(z)+\left(z^{-1}\right)^{*}M(z)^{*}\right)\geq c\label{eq:pos_def}
\end{equation}
holds. For $\rho>\rho_{0}$ we denote the canonical extension of $A$
to $H_{\rho,0}(\mathbb{R},H)$ again by $A$. Then the evolutionary
problem 
\[
\overline{\left(\partial_{0}M\left(\partial_{0}^{-1}\right)+A\right)}U=F
\]
is well-posed in the sense that $\overline{\left(\partial_{0}M\left(\partial_{0}^{-1}\right)+A\right)}$
has a bounded inverse on $H_{\rho,0}(\mathbb{R},H)$. Moreover, the
inverse is causal.
\end{thm}
For the models under consideration it suffices to consider $M$ as
a rational, bounded-operator-valued function, which, possibly by eliminating
removable singularities, is analytic at $0$ (in \citep{Picard2009}
these material laws are called $0-$analytic). This means in particular
that $M$ can be factorized in the form
\begin{equation}
M\left(z\right)=\prod_{k=0}^{s}Q_{k}\left(z\right)^{-1}P_{k}\left(z\right)\label{eq:rational0-1}
\end{equation}
where $P_{k},Q_{k}$ are polynomials.%
\footnote{The form $U=M\left(\partial_{0}^{-1}\right)V$ may be interpreted
as coming from solving an integro-differential equation of the form
\[
\partial_{0}^{N}Q\left(\partial_{0}^{-1}\right)U=\partial_{0}^{N}P\left(\partial_{0}^{-1}\right)V,
\]
where $N\in\mathbb{N}$ is the degree of the operator polynomial $Q$.%
} In this case, condition (\ref{eq:pos_def}) simplifies to
\begin{equation}
\rho M\left(0\right)+\Re M^{\prime}\left(0\right)\geq c\label{eq:pos-def2-1}
\end{equation}
for some $c>0$ and all sufficiently large $\rho>0$. Indeed, the
only difference between the expression in \eqref{eq:pos-def2-1} and
\eqref{eq:pos_def} are terms multiplied by a multiple of $\left|z\right|=\left|\frac{1}{it+\rho}\right|$,
which are eventually small, if $\rho>0$ is chosen sufficiently large.
A finer classification of these models can be obtained by looking
at the (unbounded) linear operator $A$ and the ``zero patterns''
of $M\left(0\right)$ and $\mathrm{\Re}M^{\prime}\left(0\right)$.
\begin{comment}
In addition we may have higher order terms appearing in $M$ or not.
\end{comment}

\subsection{The equations of thermo-elasticity}

We start with the classical equations of irreversible thermo-elasticity
in an elastic body $\Omega\subseteq\mathbb{R}^{3}$ due to Biot \citep{Biot1956}.
Before we can formulate these equations properly, we need to define
the spatial differential operators involved. 
\begin{defn}
We define the operator $\interior\grad$ as the closure of 
\begin{align*}
\grad|_{\interior C_{\infty}(\Omega)}:\interior C_{\infty}(\Omega)\subseteq L^{2}(\Omega) & \to L^{2}(\Omega)^{3}\\
\phi & \mapsto\left(\partial_{1}\phi,\partial_{2}\phi,\partial_{3}\phi\right),
\end{align*}
where we recall that $\interior C_{\infty}(\Omega)$ denotes the space
of smooth functions with compact support in $\Omega.$ Likewise we
define $\interior\dive$ as the closure of 
\begin{align*}
\dive|_{\interior C_{\infty}(\Omega)^{3}}:\interior C_{\infty}(\Omega)^{3}\subseteq L^{2}(\Omega)^{3} & \to L^{2}(\Omega)\\
(\phi_{1},\phi_{2},\phi_{3}) & \mapsto\sum_{i=1}^{3}\partial_{i}\phi_{i}.
\end{align*}
Integration by parts yield $\interior\grad\subseteq-\left(\interior\dive\right)^{\ast}\eqqcolon\grad$
and, similarly, $\interior\dive\subseteq-\left(\interior\grad\right)^{\ast}\eqqcolon\dive.$
\\
Moreover, we define the operator 
\begin{align*}
\sym:L^{2}(\Omega)^{3\times3} & \to L^{2}(\Omega)^{3\times3}\\
\Phi & \mapsto\left(x\mapsto\frac{1}{2}\left(\Phi(x)+\Phi(x)^{\top}\right)\right),
\end{align*}
which clearly is the orthogonal projector onto the closed subspace
\[
L_{\mathrm{sym}}^{2}(\Omega)^{3\times3}\coloneqq\left\{ \Phi\in L^{2}(\Omega)^{3\times3}\,|\,\Phi(x)=\Phi(x)^{\top}\quad(x\in\Omega\mbox{ a.e.})\right\} 
\]
 of $L^{2}(\Omega)^{3\times3}.$ Similar to the definition above,
we define the operator $\interior\Grad$ as the closure of 
\begin{align*}
\Grad|_{\interior C_{\infty}(\Omega)^{3}}:\interior C_{\infty}(\Omega)^{3}\subseteq L^{2}(\Omega)^{3} & \to L_{\mathrm{sym}}^{2}(\Omega)^{3\times3}\\
(\phi_{1},\phi_{2},\phi_{3}) & \mapsto\left(\frac{1}{2}\left(\partial_{j}\phi_{i}+\partial_{i}\phi_{j}\right)\right)_{i,j\in\{1,2,3\}}
\end{align*}
and $\interior\Dive$ as the closure of 
\begin{align*}
\Dive|_{\sym[\interior C_{\infty}(\Omega)^{3\times3}]}:\sym[\interior C_{\infty}(\Omega)^{3\times3}]\subseteq L_{\mathrm{sym}}^{2}(\Omega)^{3\times3} & \to L^{2}(\Omega)^{3}\\
\left(\phi_{ij}\right)_{i,j\in\{1,2,3\}} & \mapsto\left(\sum_{j=1}^{3}\partial_{j}\phi_{ij}\right)_{i\in\{1,2,3\}}.
\end{align*}
By integration by parts we again obtain $\interior\Grad\subseteq-\left(\interior\Dive\right)^{\ast}\eqqcolon\Grad$
as well as $\interior\Dive\subseteq-\left(\interior\Grad\right)^{\ast}\eqqcolon\Dive$. 
\end{defn}
We are now able to formulate the equations of thermo-elasticity. Let
$u\in H_{\rho,0}(\mathbb{R},L^{2}(\Omega)^{3})$ denote the displacement-field
of the elastic body $\Omega$ and $\sigma\in H_{\rho,0}(\mathbb{R},L_{\mathrm{sym}}^{2}(\Omega)^{3\times3})$
the stress. Then $u$ and $\sigma$ satisfy the balance of momentum
equation 
\begin{equation}
\rho_{0}\partial_{0}^{2}u-\Dive\sigma=f,\label{eq:elast}
\end{equation}
where $\rho_{0}\in L^{\infty}(\Omega)$ denotes the mass-density of
$\Omega$ and $f\in H_{\rho,0}(\mathbb{R},L^{2}(\Omega)^{3})$ is
an external forcing term. Furthermore, let $\eta\in H_{\rho,0}(\mathbb{R},L^{2}(\Omega))$
denote the entropy and $q\in H_{\rho,0}(\mathbb{R},L^{2}(\Omega)^{3})$
the heat flux. Then, these quantities satisfy the conservation law
\begin{equation}
\rho_{0}\partial_{0}\eta+\dive(T_{0}^{-1}q)=T_{0}^{-1}h,\label{eq:heat}
\end{equation}
where $T_{0}$ denotes the reference temperature%
\footnote{For simplicity we have set the reference temperature $T_{0}$ in the
introduction (and also later on) to $T_{0}=1$. In equation \eqref{eq:heat}
we let $T_{0}\in\oi0\infty$ be arbitrary to keep the formulation
more comparable with the classically proposed models.%
} and $h\in H_{\rho,0}(\mathbb{R},L^{2}(\Omega))$ is a heating source
term. The equations are completed by the following relations
\begin{align}
\sigma & =C\varepsilon-\Gamma\theta,\label{eq:Hooke}\\
\rho_{0}\eta & =\Gamma^{\ast}\varepsilon+\nu\theta,\label{eq:entropy}\\
q & =-\kappa\grad\theta.\label{eq:Fourier}
\end{align}
Here $\varepsilon=\Grad u$ is the strain, $\theta\in H_{\rho,0}(\mathbb{R},L^{2}(\Omega))$
denotes the temperature, $C\in L(L_{\mathrm{sym}}^{2}(\Omega)^{3\times3})$
is the elasticity tensor, $\nu\in L^{\infty}(\Omega)$ stands for
the specific heat, $\kappa\in L^{\infty}(\Omega)$ denotes the thermal
conductivity and $\Gamma\in L(L^{2}(\Omega),L_{\mathrm{sym}}^{2}(\Omega)^{3\times3})$
is the thermo-elasticity tensor that results from the Duhamel-Neumann
law linking the stress to strain and temperature. Assuming that $C$
is invertible, we may rewrite (\ref{eq:Hooke}) as 
\begin{equation}
\varepsilon=C^{-1}\sigma+C^{-1}\Gamma\theta.\label{eq:mod-Hooke-2}
\end{equation}
Consequently, (\ref{eq:entropy}) can be written as 
\begin{equation}
\rho_{0}\eta=\Gamma^{\ast}C^{-1}\sigma+\left(\Gamma^{\ast}C^{-1}\Gamma+\nu\right)\theta\label{eq:mod-entropy}
\end{equation}
and, hence, with $v\coloneqq\partial_{0}u,\sigma,\theta$ and $q$
as our basic unknowns, (\ref{eq:elast}),(\ref{eq:Fourier}),(\ref{eq:mod-Hooke-2})
and (\ref{eq:mod-entropy}) can be combined to the following equations
on $H_{\rho,0}(\mathbb{R},H),$ where $H\coloneqq L^{2}(\Omega)^{3}\oplus L_{\mathrm{sym}}^{2}(\Omega)^{3\times3}\oplus L^{2}(\Omega)\oplus L^{2}(\Omega)^{3}$:
\begin{multline*}
\left(\partial_{0}\left(\begin{array}{cccc}
\rho_{0} & 0 & 0 & 0\\
0 & C^{-1} & C^{-1}\Gamma & 0\\
0 & \Gamma^{\ast}C^{-1} & \Gamma^{\ast}C^{-1}\Gamma+\nu & 0\\
0 & 0 & 0 & 0
\end{array}\right)+\left(\begin{array}{cccc}
0 & 0 & 0 & 0\\
0 & 0 & 0 & 0\\
0 & 0 & 0 & 0\\
0 & 0 & 0 & \kappa^{-1}
\end{array}\right)\right.\\
\left.+\left(\begin{array}{cccc}
0 & -\Dive & 0 & 0\\
-\Grad & 0 & 0 & 0\\
0 & 0 & 0 & \dive\\
0 & 0 & \grad & 0
\end{array}\right)\right)\left(\begin{array}{c}
v\\
\sigma\\
\theta\\
q
\end{array}\right)=\left(\begin{array}{c}
f\\
0\\
h\\
0
\end{array}\right).
\end{multline*}

This systems is, at least formally, of the form (\ref{eq:evo2-1}),
where $M(\partial_{0}^{-1})=M_{0}+\partial_{0}^{-1}M_{1}$ with 
\[
M_{0}=\left(\begin{array}{cccc}
\rho_{0} & 0 & 0 & 0\\
0 & C^{-1} & C^{-1}\Gamma & 0\\
0 & \Gamma^{\ast}C^{-1} & \Gamma^{\ast}C^{-1}\Gamma+\nu & 0\\
0 & 0 & 0 & 0
\end{array}\right),\; M_{1}=\left(\begin{array}{cccc}
0 & 0 & 0 & 0\\
0 & 0 & 0 & 0\\
0 & 0 & 0 & 0\\
0 & 0 & 0 & \kappa^{-1}
\end{array}\right)
\]
and 
\begin{equation}
A=\left(\begin{array}{cccc}
0 & -\Dive & 0 & 0\\
-\Grad & 0 & 0 & 0\\
0 & 0 & 0 & \dive\\
0 & 0 & \grad & 0
\end{array}\right).\label{eq:A_without}
\end{equation}
To make $A$ become skew-selfadjoint, we need to impose boundary conditions
on our unknowns. For instance, one could require homogeneous Dirichlet-conditions
for $v$ and $\theta,$ which can be formulated by $v\in D(\interior\Grad)$
and $\theta\in D(\interior\grad).$ Then, $A$ becomes 
\begin{equation}
A=\left(\begin{array}{cccc}
0 & -\Dive & 0 & 0\\
-\interior\Grad & 0 & 0 & 0\\
0 & 0 & 0 & \dive\\
0 & 0 & \interior\grad & 0
\end{array}\right),\label{eq:A}
\end{equation}
which clearly is skew-selfadjoint. Of course, other boundary conditions
can be imposed making $A$ skew-selfadjoint, see e.g.~\citep{Trostorff2014}. 

As we shall see, the Lord-Shulman model {[}2{]}, the two dual-phase
lag models {[}11, 12{]} and the three Green-Naghdy models {[}7-9{]}
are based on the same relations (\ref{eq:Hooke}), (\ref{eq:entropy}),
differences only appearing in the modification of Fourier's law (\ref{eq:Fourier}).
In the case of the Green-Lindsay model {[}6{]}, although of the same
formal shape, the meaning of the temperature $\theta$ is replaced
by the differential expression $\left(1+n_{0}\partial_{0}\right)$
applied to temperature. Therefore, in order to avoid confusion, we
shall use in this case 
\[
\Theta\coloneqq\theta+n_{0}\partial_{0}\theta
\]
instead of re-dedicating the symbol $\theta$, where $n_{0}$is the
thermal relaxation time, a characteristic of this model.

\section{Solution theory to some thermo-elastic models}

In this section we will show that the models of thermo-elasticity
mentioned in the introduction can be written as evolutionary problems
in the sense of Section 1.1 and thus, their well-posedness can be
shown with the help of Theorem \ref{thm:sol_th}. In fact, we will
show that a generalized model of the basic Green-Lindsay type allows
to recover all other models as special cases. We will begin to formulate
this abstract model and prove its well-posedness. In the subsequent
subsection, we will show how the classical models can be recovered
from the abstract one and which conditions yield their well-posedness.

\subsection{A general rational material law for thermo-elasticity}

We consider the following material law $M(\partial_{0}^{-1})=M_{0}+\partial_{0}^{-1}M_{1}(\partial_{0}^{-1}),$
where 
\begin{align}
M_{0} & =\left(\begin{array}{cccc}
\rho_{0} & 0 & 0 & 0\\
0 & C^{-1} & C^{-1}\Gamma & 0\\
0 & \quad\Gamma^{*}C^{-1} & \quad\nu+\Gamma^{*}C^{-1}\Gamma+\zeta_{0}^{*}a_{0}\zeta_{0} & \;\zeta_{0}^{*}a_{0}\\
0 & 0 & a_{0}\zeta_{0} & a_{0}
\end{array}\right),\label{eq:M_0-rational}\\
M_{1}(\partial_{0}^{-1}) & =\left(\begin{array}{cccc}
0 & 0 & 0 & 0\\
0 & 0 & 0 & 0\\
0 & 0 & a_{1}(\partial_{0}^{-1}) & 0\\
0 & 0 & 0 & a_{2}(\partial_{0}^{-1})
\end{array}\right).\label{eq:M_1-rational}
\end{align}
Here $a_{0}\in L(L^{2}(\Omega)^{3})$ is a selfadjoint operator, $\zeta_{0}\in L(L^{2}(\Omega),L^{2}(\Omega)^{3})$
and $a_{1}:B(0,r)\to L(L^{2}(\Omega))$ and $a_{2}:B(0,r)\to L(L^{2}(\Omega)^{3})$
are rational functions for some $r>0.$ We recall from the previous
section that $\rho_{0},\nu\in L^{\infty}(\Omega)$ denote the mass
density and the specific heat, respectively, which will be assumed
to be real and strictly positive, i.e. $\rho_{0}(x),\nu(x)\geq c$
for some $c>0$ and almost every $x\in\Omega.$ Moreover, the elasticity
tensor $C\in L(L_{\mathrm{sym}}^{2}(\Omega)^{3\times3})$ is assumed
to be selfadjoint and strictly positive definite. 
\begin{thm}
\label{thm:well_posedness_thermoelast}Let $M_{0}$ and $M_{1}(\partial_{0}^{-1})$
be as in \eqref{eq:M_0-rational} and \eqref{eq:M_1-rational}, respectively.
We assume that $\rho_{0},\nu\in L^{\infty}(\Omega)$ are real-valued
and strictly positive and $C\in L(L_{\mathrm{sym}}^{2}(\Omega)^{3\times3})$
is selfadjoint and strictly positive definite. Moreover, we assume
that $a_{0}$ is strictly positive definite on its range and $\Re a_{2}(0)$
is strictly positive definite on the kernel of $a_{0}$. Then the
evolutionary problem 
\begin{equation}
\left(\partial_{0}M_{0}+M_{1}(\partial_{0}^{-1})+A\right)\left(\begin{array}{c}
v\\
\sigma\\
\Theta\\
q
\end{array}\right)=\left(\begin{array}{c}
f\\
0\\
h\\
0
\end{array}\right)\label{eq:thermo_abstract}
\end{equation}
is well-posed in the sense of Theorem \ref{thm:sol_th}, where $A$
is given by%
\footnote{Or any other skew-selfadjoint restriction of (\ref{eq:A_without}).%
} (\ref{eq:A}) and $\Theta=\left(1+n_{0}\partial_{0}\right)\theta$.\end{thm}
\begin{proof}
According to Theorem \ref{thm:sol_th}, we need to verify condition
\eqref{eq:pos_def} for $M(\partial_{0}^{-1})=M_{0}+\partial_{0}^{-1}M_{1}(\partial_{0}^{-1})$.
Or, equivalently, by the structural properties assumed for $M_{1}$,
we need to verify that there exists $\rho_{1}>0$ such that for all
$\rho>\rho_{1}$ we have that
\[
\rho M_{0}+\Re M_{1}(0)\geq c
\]
 for some $c>0$. Indeed, the latter equation is precisely the reformulation
of \eqref{eq:pos-def2-1} for the particular $M$ under consideration.
For $\rho>0$, we compute
\begin{align*}
 & \rho M_{0}+\Re M_{1}(0)\\
 & =\rho\left(\begin{array}{cccc}
\rho_{0} & 0 & 0 & 0\\
0 & C^{-1} & C^{-1}\Gamma & 0\\
0 & \Gamma^{*}C^{-1} & \nu+\Gamma^{*}C^{-1}\Gamma+\zeta_{0}^{*}a_{0}\zeta & \zeta_{0}^{*}a_{0}\\
0 & 0 & a_{0}\zeta_{0} & a_{0}
\end{array}\right)+\left(\begin{array}{cccc}
0 & 0 & 0 & 0\\
0 & 0 & 0 & 0\\
0 & 0 & \Re a_{1}(0) & 0\\
0 & 0 & 0 & \Re a_{2}(0)
\end{array}\right)\\
 & =\left(\begin{array}{cccc}
1 & 0 & 0 & 0\\
0 & 1 & 0 & 0\\
0 & \Gamma^{*} & 1 & \zeta_{0}^{*}\\
0 & 0 & 0 & 1
\end{array}\right)\times\\
 & \qquad\left(\rho\left(\begin{array}{cccc}
\rho_{0} & 0 & 0 & 0\\
0 & C^{-1} & 0 & 0\\
0 & 0 & \nu & 0\\
0 & 0 & 0 & a_{0}
\end{array}\right)+\left(\begin{array}{cccc}
0 & 0 & 0 & 0\\
0 & 0 & 0 & 0\\
0 & 0 & \Re a_{1}(0)+\zeta_{0}^{*}\Re a_{2}(0)\zeta_{0} & -\zeta_{0}^{*}\Re a_{2}(0)\\
0 & 0 & -\Re a_{2}(0)\zeta_{0} & \Re a_{2}(0)
\end{array}\right)\right)\left(\begin{array}{cccc}
1 & 0 & 0 & 0\\
0 & 1 & \Gamma & 0\\
0 & 0 & 1 & 0\\
0 & 0 & \zeta_{0} & 1
\end{array}\right).
\end{align*}
We read off that the latter is strictly positive definite if and only
if the operator
\[
\rho\left(\begin{array}{cccc}
\rho_{0} & 0 & 0 & 0\\
0 & C^{-1} & 0 & 0\\
0 & 0 & \nu & 0\\
0 & 0 & 0 & a_{0}
\end{array}\right)+\left(\begin{array}{cccc}
0 & 0 & 0 & 0\\
0 & 0 & 0 & 0\\
0 & 0 & \Re a_{1}(0)+\zeta_{0}^{*}\Re a_{2}(0)\zeta_{0} & -\zeta_{0}^{*}\Re a_{2}(0)\\
0 & 0 & -\Re a_{2}(0)\zeta_{0} & \Re a_{2}(0)
\end{array}\right)
\]
is strictly positive definite. As, by assumption, the operators $\rho_{0}$
and $C^{-1}$ are positive definite anyway, we only have to study
the positive definiteness of the operator
\begin{equation}
\rho\left(\begin{array}{cc}
\nu & 0\\
0 & a_{0}
\end{array}\right)+\left(\begin{array}{cc}
\Re a_{1}(0)+\zeta_{0}^{*}\Re a_{2}(0)\zeta_{0} & -\zeta_{0}^{*}\Re a_{2}(0)\\
-\Re a_{2}(0)\zeta_{0} & \Re a_{2}(0)
\end{array}\right).\label{eq:Rest of the operator}
\end{equation}
Now, decomposing the underlying (spatial) Hilbert space as 
\[
L^{2}(\Omega)\oplus L^{2}(\Omega)^{3}=R\left(\left(\begin{array}{cc}
\nu & 0\\
0 & a_{0}
\end{array}\right)\right)\oplus N\left(\left(\begin{array}{cc}
\nu & 0\\
0 & a_{0}
\end{array}\right)\right),
\]
 which can be done, since both $a_{0}$ and $\nu$ are strictly positive
definite on the respective ranges, we realize that \eqref{eq:Rest of the operator}
is strictly positive definite on $H_{\rho,0}\left(\mathbb{R},R\left(\left(\begin{array}{cc}
\nu & 0\\
0 & a_{0}
\end{array}\right)\right)\right)$ with positive definiteness constant arbitrarily large, depending
on the choice of $\rho>0.$ By Euklid's inequality ($2ab\leq\frac{1}{\epsilon}a^{2}+\epsilon b^{2},$
$a,b\in\mathbb{R}$, $\epsilon>0$), the assertion follows, if we
show that the operator in \eqref{eq:Rest of the operator} is strictly
positive definite on the nullspace of $a_{0}$. However, by assumption,
$\Re a_{2}(0)$ is strictly positive on $N(a_{0})$. This yields the
assertion.\end{proof}
\begin{rem}
We write down Equation (\ref{eq:thermo_abstract}) line by line. It
is 
\begin{align*}
\partial_{0}\rho_{0}v-\Dive\sigma & =f,\\
\partial_{0}C^{-1}\sigma+\partial_{0}C^{-1}\Gamma\Theta-\interior\Grad v & =0,\\
\partial_{0}\Gamma^{\ast}C^{-1}\sigma+\partial_{0}\left(\nu+\Gamma^{\ast}C^{-1}\Gamma+\zeta_{0}^{\ast}a_{0}\zeta_{0}\right)\Theta+\partial_{0}\zeta_{0}^{\ast}a_{0}q+a_{1}(\partial_{0}^{-1})\Theta+\dive q & =h,\\
\partial_{0}a_{0}\zeta_{0}\Theta+\partial_{0}a_{0}q+a_{2}(\partial_{0}^{-1})q+\interior\grad\Theta & =0.
\end{align*}
Defining $u\coloneqq\partial_{0}^{-1}v,\varepsilon\coloneqq\Grad u$
and $\eta\coloneqq\rho_{0}^{-1}\left(\Gamma^{\ast}\varepsilon+\left(\nu+\zeta_{0}^{\ast}a_{0}\zeta_{0}\right)\Theta+\zeta_{0}^{\ast}a_{0}q+\partial_{0}^{-1}a_{1}(\partial_{0}^{-1})\Theta\right)$
we get from the second line 
\[
\sigma=C\varepsilon-\Gamma\Theta.
\]
Moreover, the fourth line reads as 
\[
\partial_{0}a_{0}q+a_{2}(\partial_{0}^{-1})q=-\partial_{0}a_{0}\zeta_{0}\Theta-\interior\grad\Theta
\]
and the first line is 
\[
\partial_{0}^{2}\rho_{0}u-\Dive\sigma=f.
\]
Finally, the third line reads as 
\begin{align*}
\partial_{0}\rho_{0}\eta+\dive q & =\partial_{0}\left(\Gamma^{\ast}\varepsilon+\left(\nu+\zeta_{0}^{\ast}a_{0}\zeta_{0}\right)\Theta+\zeta_{0}^{\ast}a_{0}q+\partial_{0}^{-1}a_{1}(\partial_{0}^{-1})\Theta\right)+a_{1}(\partial_{0}^{-1})\Theta+\dive q\\
 & =\partial_{0}\Gamma^{\ast}C^{-1}\sigma+\partial_{0}\left(\Gamma^{\ast}C^{-1}\Gamma\Theta+\nu+\zeta_{0}^{\ast}a_{0}\zeta_{0}\right)\Theta+\partial_{0}\zeta_{0}^{\ast}a_{0}q+a_{1}(\partial_{0}^{-1})\Theta+\dive q\\
 & =h,
\end{align*}
where we have used $\sigma=C\varepsilon-\Gamma\Theta$. Summarizing,
our material relations are 
\begin{align}
\partial_{0}^{2}\rho_{0}u-\Dive\sigma & =f,\label{eq:elast_new}\\
\partial_{0}\rho_{0}\eta+\dive q & =h,\label{eq:heat_new}\\
\sigma & =C\varepsilon-\Gamma\Theta,\label{eq:Hooke_new}\\
\rho_{0}\eta & =\Gamma^{\ast}\varepsilon+\left(\nu+\zeta_{0}^{\ast}a_{0}\zeta_{0}\right)\Theta+\zeta_{0}^{\ast}a_{0}q+\partial_{0}^{-1}a_{1}(\partial_{0}^{-1})\Theta,\label{eq:entropy_new}\\
\partial_{0}a_{0}q+a_{2}(\partial_{0}^{-1})q & =-\partial_{0}a_{0}\zeta_{0}\Theta-\interior\grad\Theta,\label{eq:Fourier_new}\\
\Theta & =(1+n_{0}\partial_{0})\theta.\label{eq:Theta_new}
\end{align}

We note that for $n_{0}=a_{0}=\zeta_{0}=a_{1}(\partial_{0}^{-1})=0$
and $a_{2}(\partial_{0}^{-1})=\kappa^{-1}$ we recover the equations
of irreversible thermo-elasticity (compare Subsection 1.2).
\end{rem}
We will now discuss several models of thermo-elasticity and we will
show that they all are covered by the model proposed above. Due to
the importance of $M\left(0\right)=M_{0}$, $M^{\prime}\left(0\right)=M_{1}(0)$
in the discussion of well-posedness, compare (\ref{eq:pos-def2-1}),
we are first lead to distinguish two classes of models. 
\begin{itemize}
\item Generic models.

These models are characterized by $M\left(0\right)=\mathrm{\Re}M\left(0\right)$
being strictly positive definite. For these (\ref{eq:pos-def2-1})
is always satisfied. Moreover, $M\left(0\right)+\epsilon$ is then
also strictly positive definite for any sufficiently small selfadjoint
operator $\epsilon$. 

\item Degenerate models.

These models fail to have the remarkable stability with regards to
perturbations of the generic models. They are characterized by $M\left(0\right)=\mathrm{\Re}M\left(0\right)$
having a non-trivial null space. In these cases (\ref{eq:pos-def2-1})
can be ensured for example by assuming that $M\left(0\right)=\Re M\left(0\right)$
is strictly positive definite on its own range $M\left(0\right)\left[H\right]$,
i.e.
\[
\left\langle x|M\left(0\right)x\right\rangle _{H}\geq c_{0}>0\mbox{ for all }x\in M\left(0\right)\left[H\right],
\]
and $\Re M^{\prime}\left(0\right)$ being strictly positive definite
on the null space $\left[\left\{ 0\right\} \right]M\left(0\right)$,
i.e.
\[
\left\langle x|\Re M^{\prime}\left(0\right)x\right\rangle _{H}\geq c_{0}>0\mbox{ for all }x\in\left[\left\{ 0\right\} \right]M\left(0\right).
\]

\end{itemize}

\subsection{The generic case}

\subsubsection{Lord-Shulman model}

In contrast to the model for irreversible thermo-elasticity (compare
Subsection 1.2), Lord and Shulman (\citep{Lord1967}) proposed to
replace Fourier's law (\ref{eq:Fourier}) by the so-called Cattaneo
modification of Fourier's law (see \citep{cattaneo1958}), which reads
as
\begin{align*}
\partial_{0}a_{0}q+q & =-\kappa\interior\grad\theta,
\end{align*}
where $a_{0}\in L^{\infty}(\Omega)$ is assumed to be real-valued
and strictly positive definite. This results in a system of the form
(\ref{eq:thermo_abstract}), where 
\begin{equation}
n_{0}=\zeta_{0}=a_{1}(\partial_{0}^{-1})=0\text{ and }a_{2}(\partial_{0}^{-1})=\kappa^{-1}.\label{eq:zeros_LSm}
\end{equation}
In consequence, we obtain the well-posedness for this model by Theorem
\ref{thm:well_posedness_thermoelast}:
\begin{cor}
\label{cor:Lord-Shulman}Let $M_{0}$, $M_{1}(\partial_{0}^{-1})$
and $A$ be given by \eqref{eq:M_0-rational}, \eqref{eq:M_1-rational},
and a skew-selfadjoint restriction of \eqref{eq:A_without}, respectively.
Assume that $a_{0},\rho_{0},\nu\in L^{\infty}(\Omega),$ $\kappa\in L(L^{2}(\Omega)^{3})$,
$C\in L(L_{\textnormal{sym}}^{2}(\Omega)^{3\times3})$ are selfadjoint
and strictly positive definite%
\footnote{For ease of formulation, note that we identified $a_{0}$, $\rho_{0}$
and $\nu$ with the induced multiplication operators on $L^{2}$.
In this way, selfadjointness is just the same as to say the respective
$L^{\infty}$-functions assume only real values and, thus, strict
positivity coincides with strict positivity of the respective functions.%
} as well as \eqref{eq:zeros_LSm}. Then \eqref{eq:thermo_abstract}
is well-posed in the sense of Theorem \ref{thm:sol_th}.\end{cor}
\begin{proof}
It suffices to observe that the kernel of $a_{0}$ is trivial.\end{proof}
\begin{rem}
If we mark possible non-zero entries in the operator matrix $M\left(0\right)$
by a star, we have the zero-pattern
\[
M\left(0\right)=\left(\begin{array}{cccc}
\bigstar & 0 & 0 & 0\\
0 & \bigstar & \bigstar & 0\\
0 & \bigstar & \bigstar & 0\\
0 & 0 & 0 & \bigstar
\end{array}\right).
\]
The zero-pattern of $M_{1}$ is
\[
M_{1}(\partial_{0}^{-1})=M_{1}(0)=\left(\begin{array}{cccc}
0 & 0 & 0 & 0\\
0 & 0 & 0 & 0\\
0 & 0 & 0 & 0\\
0 & 0 & 0 & \bigstar
\end{array}\right).
\]
Hence, we observe that there are no higher order terms in the material
law operator in case of Lord-Shulman model. 
\end{rem}

\subsubsection{Green-Naghdy model of type II}

In the models proposed by Green and Naghdy (see \citep{Green_Naghdi1991,Green_Naghdi1992,Green_Naghdi1993})
a modified heat flux of the form 
\begin{align*}
q & =-\partial_{0}^{-1}\left(k^{\ast}+k\partial_{0}\right)\interior\grad\theta\\
 & =-(\partial_{0}^{-1}k^{\ast}+k)\interior\grad\theta.
\end{align*}
is assumed by considering \emph{$k,k^{\ast}\in\mathbb{R}$} as thermal
conductivity and conductivity rate, respectively. Depending on $k^{\ast}$
and $k$, we distinguish between three types of this model. If $k=0$
and $k^{\ast}\ne0$, we speak about the Green-Naghdy Model of Type
II. In this case, the above heat flux satisfies 
\[
\partial_{0}\left(k^{\ast}\right)^{-1}q=-\interior\grad\theta.
\]
This system is covered by the abstract one if 
\begin{equation}
n_{0}=\zeta_{0}=a_{1}(\partial_{0}^{-1})=a_{2}(\partial_{0}^{-1})=0\text{ and }a_{0}=\left(k^{\ast}\right)^{-1}.\label{eq:zeros_GNII}
\end{equation}
Hence, $M(0)$ has the zero-pattern
\[
M(0)=\left(\begin{array}{cccc}
\bigstar & 0 & 0 & 0\\
0 & \bigstar & \bigstar & 0\\
0 & \bigstar & \bigstar & 0\\
0 & 0 & 0 & \bigstar
\end{array}\right),
\]
while 
\[
M_{1}(\partial_{0}^{-1})=\left(\begin{array}{cccc}
0 & 0 & 0 & 0\\
0 & 0 & 0 & 0\\
0 & 0 & 0 & 0\\
0 & 0 & 0 & 0
\end{array}\right).
\]
The corresponding well-posedness result in a generalized fashion reads
as:
\begin{cor}
Let $M_{0}$, $M_{1}(\partial_{0}^{-1})$ and $A$ be given by \eqref{eq:M_0-rational},
\eqref{eq:M_1-rational}, and a skew-selfadjoint restriction of \eqref{eq:A_without},
respectively. Assume that $k^{*},\rho_{0},\nu\in L^{\infty}(\Omega),$
$C\in L(L_{\textnormal{sym}}^{2}(\Omega)^{3\times3})$ are selfadjoint
and strictly positive definite as well as \eqref{eq:zeros_GNII}.
Then \eqref{eq:thermo_abstract} is well-posed in the sense of Theorem
\ref{thm:sol_th}.\end{cor}
\begin{proof}
Again, the assertion follows when applying Theorem \ref{thm:well_posedness_thermoelast}
while observing that $N(a_{0})=\{0\}$.
\end{proof}

\subsubsection{The generic Green-Lindsay model}

As indicated earlier in Section 1.2, here the material relations (\ref{eq:Hooke})-(\ref{eq:Fourier})
are modified to 
\begin{align*}
\sigma & =C\epsilon-\Gamma\left(\theta+n_{0}\partial_{0}\theta\right),\\
\rho_{0}\eta & =d\theta+h\partial_{0}\theta+\Gamma^{*}\epsilon-b^{\ast}\interior\grad\theta,\\
q & =-b\partial_{0}\theta-\kappa\interior\grad\theta.
\end{align*}
Here \emph{$b,d$} are material parameters and $h,n_{0}(\ne0)$ are
the thermal relaxation times. Now, letting 
\[
\Theta\coloneqq\theta+n_{0}\partial_{0}\theta
\]
we get
\begin{align*}
\theta & =\left(1+n_{0}\partial_{0}\right)^{-1}\Theta\\
 & =\partial_{0}^{-1}\left(\partial_{0}^{-1}+n_{0}\right)^{-1}\Theta
\end{align*}
and the material relations above turn into
\begin{align*}
\epsilon & =C^{-1}\sigma+C^{-1}\Gamma\Theta,\\
q & =-\left(\partial_{0}^{-1}+n_{0}\right)^{-1}\left(b\Theta+\kappa\partial_{0}^{-1}\interior\grad\Theta\right),
\end{align*}
which yields, 
\[
\partial_{0}n_{0}\kappa^{-1}q+\kappa^{-1}q=-\partial_{0}\kappa^{-1}b\Theta-\interior\grad\Theta.
\]
Moreover, we have, using the Neumann series, 
\begin{align*}
\rho_{0}\eta & =d\theta+h\partial_{0}\theta+\Gamma^{*}\varepsilon-b^{\ast}\interior\grad\theta\\
 & =\left(d+h\partial_{0}\right)\partial_{0}^{-1}\left(\partial_{0}^{-1}+n_{0}\right)^{-1}\Theta+\Gamma^{\ast}\varepsilon+b^{\ast}\kappa^{-1}b\left(\partial_{0}^{-1}+n_{0}\right)^{-1}\Theta+b^{\ast}\kappa^{-1}q\\
 & =\Gamma^{\ast}\varepsilon+\left(\partial_{0}^{-1}d+h\right)n_{0}^{-1}\sum_{j=0}^{\infty}\left(-\partial_{0}^{-1}n_{0}^{-1}\right)^{j}\Theta+b^{\ast}\kappa^{-1}bn_{0}^{-1}\sum_{j=0}^{\infty}(-\partial_{0}^{-1}n_{0}^{-1})^{j}\Theta+b^{\ast}\kappa^{-1}q\\
 & =\Gamma^{\ast}\varepsilon+\left(hn_{0}^{-1}+b^{\ast}\kappa^{-1}bn_{0}^{-1}\right)\Theta+b^{\ast}\kappa^{-1}q\\
 & \qquad+\partial_{0}^{-1}\left(dn_{0}^{-1}-\left(h+b^{\ast}\kappa^{-1}b\right)n_{0}^{-2}\right)\sum_{j=0}^{\infty}\left(-\partial_{0}^{-1}n_{0}^{-1}\right)^{j}\Theta\\
 & =\Gamma^{\ast}\varepsilon+\left(hn_{0}^{-1}+b^{\ast}\kappa^{-1}bn_{0}^{-1}\right)\Theta+b^{\ast}\kappa^{-1}q+\partial_{0}^{-1}\left(d-\left(h+b^{\ast}\kappa^{-1}b\right)n_{0}^{-1}\right)\left(n_{0}+\partial_{0}^{-1}\right)^{-1}\Theta.
\end{align*}
Thus, we are in our abstract situation with 
\begin{align}
a_{0} & =n_{0}\kappa^{-1},a_{2}(\partial_{0}^{-1})=\kappa^{-1},\zeta_{0}=bn_{0}^{-1},\nu=hn_{0}^{-1}\text{ and }\label{eq:zeros_GLm-gen}\\
a_{1}(\partial_{0}^{-1}) & =\left(d-\left(h+b^{\ast}\kappa^{-1}b\right)n_{0}^{-1}\right)\left(n_{0}+\partial_{0}^{-1}\right)^{-1}\nonumber 
\end{align}
In this case the operator matrix $M\left(0\right)$ has the zero-pattern
\[
M\left(0\right)=\left(\begin{array}{cccc}
\bigstar & 0 & 0 & 0\\
0 & \bigstar & \bigstar & 0\\
0 & \bigstar & \bigstar & \bigstar\\
0 & 0 & \bigstar & \bigstar
\end{array}\right).
\]
The zero-pattern of $M_{1}(0)$ is now 
\[
M_{1}(0)=\left(\begin{array}{cccc}
0 & 0 & 0 & 0\\
0 & 0 & 0 & 0\\
0 & 0 & \bigstar & 0\\
0 & 0 & 0 & \bigstar
\end{array}\right).
\]
Also, we note that there are higher order terms in the material law
operator. The corresponding well-posedness result is therefore noted
in the following corollary.
\begin{cor}
Let $M_{0}$, $M_{1}(\partial_{0}^{-1})$ and $A$ be given by \eqref{eq:M_0-rational},
\eqref{eq:M_1-rational}, and a skew-selfadjoint restriction of \eqref{eq:A_without},
respectively. Let $n_{0}>0$ and assume that $h,\rho_{0}\in L^{\infty}(\Omega),$
$\kappa\in L(L^{2}(\Omega)^{3})$, $C\in L(L_{\textnormal{sym}}^{2}(\Omega)^{3\times3})$
are selfadjoint and strictly positive definite, $b\in L(L^{2}(\Omega),L^{2}(\Omega)^{3})$
as well as \eqref{eq:zeros_GLm-gen}. Then \eqref{eq:thermo_abstract}
is well-posed in the sense of Theorem \ref{thm:sol_th}.\end{cor}
\begin{proof}
Again, note that the kernel of $a_{0}$ is trivial. Apply Theorem
\ref{thm:well_posedness_thermoelast}.
\end{proof}

\subsubsection{Dual phase-lag model of type II }

In case of the model DPL-II, apart from (\ref{eq:Hooke}), (\ref{eq:entropy})
we have here the modified Fourier law as 
\begin{align*}
\left(1+n_{1}\partial_{0}+\frac{1}{2}n_{1}^{2}\partial_{0}^{2}\right)q & =-\kappa\left(1+n_{2}\partial_{0}\right)\interior\grad\theta.
\end{align*}
where $n_{1},n_{2}\in\mathbb{R}\setminus\{0\}$ are called phase-lags.
Assuming that $\kappa$ is invertible, we can write the latter relation
as

\begin{align*}
-\interior\grad\theta & =\left(1+n_{2}\partial_{0}\right)^{-1}\left(1+n_{1}\partial_{0}+\frac{1}{2}n_{1}^{2}\partial_{0}^{2}\right)\kappa^{-1}q\\
 & =\left(\partial_{0}^{-1}+n_{1}+\frac{1}{2}n_{1}^{2}\partial_{0}\right)\left(\partial_{0}^{-1}+n_{2}\right)^{-1}\kappa^{-1}q\\
 & =\frac{1}{2}n_{1}^{2}\partial_{0}\left(\partial_{0}^{-1}+n_{2}\right)^{-1}\kappa^{-1}q+\left(n_{1}+\partial_{0}^{-1}\right)\left(\partial_{0}^{-1}+n_{2}\right)^{-1}\kappa^{-1}q\\
 & =\frac{1}{2}n_{1}^{2}n_{2}^{-1}\partial_{0}\sum_{j=0}^{\infty}\left(-\partial_{0}^{-1}n_{2}^{-1}\right)^{j}\kappa^{-1}q+\left(n_{1}+\partial_{0}^{-1}\right)\left(\partial_{0}^{-1}+n_{2}\right)^{-1}\kappa^{-1}q\\
 & =\frac{1}{2}n_{1}^{2}n_{2}^{-1}\partial_{0}\kappa^{-1}q-\frac{1}{2}n_{1}^{2}n_{2}^{-2}\sum_{j=0}^{\infty}\left(-\partial_{0}^{-1}n_{2}^{-1}\right)^{j}\kappa^{-1}q+\left(n_{1}+\partial_{0}^{-1}\right)\left(\partial_{0}^{-1}+n_{2}\right)^{-1}\kappa^{-1}q\\
 & =\frac{1}{2}n_{1}^{2}n_{2}^{-1}\partial_{0}\kappa^{-1}q+\left(\left(n_{1}+\partial_{0}^{-1}\right)-\frac{1}{2}n_{1}^{2}n_{2}^{-1}\right)\left(\partial_{0}^{-1}+n_{2}\right)^{-1}\kappa^{-1}q.
\end{align*}
Thus, this corresponds to the abstract situation when 
\begin{align}
n_{0} & =\zeta_{0}=a_{1}(\partial_{0}^{-1})=0\text{ and }\label{eq:zeros_Dual_phase-lagII}\\
a_{0} & =\frac{1}{2}n_{1}^{2}n_{2}^{-1}\kappa^{-1},\, a_{2}(\partial_{0}^{-1})=\left(\left(n_{1}+\partial_{0}^{-1}\right)-\frac{1}{2}n_{1}^{2}n_{2}^{-1}\right)\left(\partial_{0}^{-1}+n_{2}\right)^{-1}\kappa^{-1}.\nonumber 
\end{align}

Therefore, the zero-pattern of $M(0)$ is
\[
M\left(0\right)=\left(\begin{array}{cccc}
\bigstar & 0 & 0 & 0\\
0 & \bigstar & \bigstar & 0\\
0 & \bigstar & \bigstar & 0\\
0 & 0 & 0 & \bigstar
\end{array}\right).
\]
and the zero-pattern of $M_{1}(0)$ is
\[
M_{1}=\Re M_{1}=\left(\begin{array}{cccc}
0 & 0 & 0 & 0\\
0 & 0 & 0 & 0\\
0 & 0 & 0 & 0\\
0 & 0 & 0 & \bigstar
\end{array}\right).
\]
It is seen that this is similar to the case of the Lord-Shulman model.
Thus, the well-posedness conditions are similar to the ones in Corollary
\ref{cor:Lord-Shulman} using \eqref{eq:zeros_Dual_phase-lagII} instead
of \eqref{eq:zeros_LSm}. However, there are (different) higher order
terms in the material law.

\subsection{The P-degenerate case}

\subsubsection{Green-Naghdy model of type- I and type- III}

Recall that in the Green-Naghdy model (see Subsection 2.2.2), Fourier's
law is replaced by 
\[
q=-(\partial_{0}^{-1}k^{\ast}+k)\interior\grad\theta.
\]
In the Green-Naghdy model of type I, it is assumed that $k^{*}=0$,
$k>0$. Thus, the above relation becomes $q=-k\interior\grad\theta,$
which is the classical Fourier law and so we have
\[
M(0)=\left(\begin{array}{cccc}
\bigstar & 0 & 0 & 0\\
0 & \bigstar & \bigstar & 0\\
0 & \bigstar & \bigstar & 0\\
0 & 0 & 0 & 0
\end{array}\right),
\]

and
\[
M_{1}(\partial_{0}^{-1})=M_{1}(0)=\left(\begin{array}{cccc}
0 & 0 & 0 & 0\\
0 & 0 & 0 & 0\\
0 & 0 & 0 & 0\\
0 & 0 & 0 & \bigstar
\end{array}\right),
\]
with no higher order terms. This is the classical model of thermoelasticity
discussed in the introduction, compare e.g. \citep{Leis1986,Picard2011,Picard2005}.
In case of the Green-Naghdy model of type III, we have that $k,k^{\ast}>0.$
This yields, that the modified Fourier law becomes 
\[
\left(\partial_{0}^{-1}k^{\ast}+k\right)^{-1}q=-\interior\grad\theta,
\]
and hence, we are in the situation of Subsection 2.1 with 
\begin{align}
n_{0} & =\zeta_{0}=a_{0}=a_{1}(\partial_{0}^{-1})=0\text{ and }\label{eq:zeros_GNm_I_and_III}\\
a_{2}(\partial_{0}^{-1}) & =\left(\partial_{0}^{-1}k^{\ast}+k\right)^{-1}.\nonumber 
\end{align}
 Thus, the zero-patterns of $M(0)$ and $M_{1}(0)$ look the same
as above with the difference that higher order terms appear (i.e.
$M_{1}(\partial_{0}^{-1})\ne M_{1}(0)$). The well-posedness result
reads as follows.
\begin{cor}
\label{cor:GNmI,III}Let $M_{0}$, $M_{1}(\partial_{0}^{-1})$ and
$A$ be given by \eqref{eq:M_0-rational}, \eqref{eq:M_1-rational},
and a skew-selfadjoint restriction of \eqref{eq:A_without}, respectively.
Assume that $\rho_{0},\nu\in L^{\infty}(\Omega),$ $k\in L(L^{2}(\Omega)^{3})$,
$C\in L(L_{\textnormal{sym}}^{2}(\Omega)^{3\times3})$ are selfadjoint
and strictly positive definite, $k^{*}\in L\left(L^{2}(\Omega)^{3}\right)$
as well as \eqref{eq:zeros_GNm_I_and_III}. Then \eqref{eq:thermo_abstract}
is well-posed in the sense of Theorem \ref{thm:sol_th}.\end{cor}
\begin{proof}
By the strict positive definiteness of $k$, it follows that $\Re a_{2}(0)=\Re k^{-1}$
is strictly positive on $L^{2}(\Omega)^{3}=N(0)=N(a_{0})$. Now, apply
Theorem \ref{thm:well_posedness_thermoelast} to obtain the required
result.
\end{proof}

\subsubsection{Dual phase-lag model of type- I }

We conclude our considerations by the study of the DPL-I\textbf{ }model.
Here again, we assume (\ref{eq:Hooke}) and (\ref{eq:entropy}) to
hold, while Fourier's law (\ref{eq:Fourier}) is replaced by 
\begin{align*}
\left(1+n_{1}\partial_{0}\right)q & =-\kappa\left(1+n_{2}\partial_{0}\right)\interior\grad\theta,
\end{align*}
with two phase-lags $n_{1},n_{2}\in\mathbb{R}\setminus\{0\}.$ The
latter gives
\begin{align*}
-\interior\grad\theta & =\left(1+n_{2}\partial_{0}\right)^{-1}\left(1+n_{1}\partial_{0}\right)\kappa^{-1}q\\
 & =\left(\partial_{0}^{-1}+n_{2}\right)^{-1}\left(\partial_{0}^{-1}+n_{1}\right)\kappa^{-1}q,
\end{align*}
which shows that we are in the case 
\begin{align}
n_{0} & =\zeta_{0}=a_{0}=a_{1}(\partial_{0}^{-1})=0\text{ and }\label{eq:zeros_Dual_I}\\
a_{2}(\partial_{0}^{-1}) & =\left(\partial_{0}^{-1}+n_{2}\right)^{-1}\left(\partial_{0}^{-1}+n_{1}\right)\kappa^{-1}.\nonumber 
\end{align}

Therefore, the zero-pattern of $M(0)$ is
\[
M\left(0\right)=\left(\begin{array}{cccc}
\bigstar & 0 & 0 & 0\\
0 & \bigstar & \bigstar & 0\\
0 & \bigstar & \bigstar & 0\\
0 & 0 & 0 & 0
\end{array}\right)
\]
and the zero-pattern of $M_{1}(0)$ is
\[
M_{1}(0)=\left(\begin{array}{cccc}
0 & 0 & 0 & 0\\
0 & 0 & 0 & 0\\
0 & 0 & 0 & 0\\
0 & 0 & 0 & \bigstar
\end{array}\right).
\]

Similarly to Corollary \ref{cor:GNmI,III}, using \eqref{eq:zeros_Dual_I}
instead of \eqref{eq:zeros_GNm_I_and_III} and imposing $n_{1}\cdot n_{2}>0$,
we get the corresponding well-posedness result also for this type
of equation.

\textbf{\noun{4. Conclusion}}

Various models of thermoelasticity are writen as evolutionary problems
and their well posedness results are shown. We formulate an abstract
model with rational material laws which is of basic Green-Lindsay
type model and we prove its well-posedness. All other models are shown
to be recovered from this abstract one and we find the conditions
which yield well posedness.

\textbf{\noun{Anowledgement: }}One of the authors (SM) thankfully
acknowledges the extended facilities provided by the Institute of
Analysis, Technical University- Dresden, Germany during the period
when the present work was carried out.

\end{document}